\documentclass[conference,letterpaper]{IEEEtran}

\usepackage[utf8]{inputenc} 
\usepackage[T1]{fontenc}
\usepackage{url}
\usepackage{ifthen}
\usepackage{cite}
\usepackage[cmex10]{amsmath} % Use the [cmex10] option to ensure complicance
\usepackage{amsthm}
\usepackage{booktabs}
\usepackage{amsfonts}
\usepackage{comment}

\usepackage{amssymb}
\usepackage[top = 0.8in, bottom=0.77in, left=0.68in, right=0.68in]{geometry}

\newtheorem{theorem}{Theorem}

\newtheorem{problem}{Problem}

\newtheorem{lemma}{Lemma}
\newtheorem{construction}{Construction}
\newtheorem{corollary}{Corollary}

\newtheorem{example}{Example}
\newtheorem{claim}{Claim}

\usepackage{xcolor}

\newcommand{\cC}{\mathcal{C}}

\newcommand{\bfc}{{\boldsymbol c}}

\newcommand{\bfu}{{\boldsymbol u}}
\newcommand{\bfv}{{\boldsymbol v}}

\title{Correcting Tail Deletions in Rank Modulated Composite Encoding for Data Storage in DNA} 

\begin{document}

\title{ Correcting Tail Deletions in Rank Modulated Composite Encoding for Data Storage in DNA}\author{\textbf{Tomer Cohen}\! \IEEEauthorrefmark{1}, \textbf{Eitan Yaakobi}\! \IEEEauthorrefmark{1}, and \textbf{Zohar Yakhini}\! \IEEEauthorrefmark{1}\IEEEauthorrefmark{2}\\
\IEEEauthorblockA{\IEEEauthorrefmark{1}Department of Computer Science, Technion - Israel Institute of Technology, Haifa 3200003, Israel}
\IEEEauthorblockA{\IEEEauthorrefmark{2}School of Computer Science, Reichman University, Herzliya, Israel}
{Emails: tomer.cohen@campus.technion.ac.il,\ yaakobi@cs.technion.ac.il,\ zohar.yakhini@gmail.com}
}
\date{\today}

\maketitle

%%%%%%
%% Abstract: 
%% If your paper is eligible for the student paper award, please add
%% the comment "THIS PAPER IS ELIGIBLE FOR THE STUDENT PAPER
%% AWARD." as a first line in the abstract. 
%% For the final version of the accepted paper, please do not forget
%% to remove this comment!
%%

\begin{abstract}
We study the combination of two recent coding approaches, in the context of DNA based data storage. Composite DNA alphabets leverage properties of the DNA synthesis and sequencing process. A composite symbol does not represent a single nucleotide, but rather a designed mixture of DNA nucleotides. Using the high multiplicity that is intrinsic to synthesis and sequencing a composite symbol consists of frequencies in the mixture. Rank modulation codes use permutations to represent information. Combining the two, we construct encoding that uses permutations of nucleotide frequencies rather than the exact frequency values.  Codes for this approach were addressed in previous work, under Kendall's tau distances. In this work we study deletion and insertion codes. We present bounds and constructions of efficient codes defined over partial permutations. 
\end{abstract}
\section{Introduction}\label{sec:intro}

The main obstacle to reducing the cost of DNA based data storage is the cost of DNA synthesis. One natural way to mitigate this cost is to increase the storage density, measured in bits per symbol or per synthesis cycle. With standard encoding over the alphabet $A,C,G,T$, the maximum achievable rate is $\log_2 4 = 2$ bits per symbol. However, practical constraints related to error correction often reduce this rate; for instance, prohibiting consecutive identical symbols limits the rate to $\log_2 3 \approx 1.58$ bits per symbol~\cite{Getal13, osti_1619517}. Expanding the set of available encoding symbols can therefore increase the achievable capacity and lower synthesis costs.

\emph{Composite DNA symbols}, proposed in~\cite{anavy_DataStorageDNA_2019,augmented_encoding}, leverage the redundancy inherent in DNA synthesis and sequencing. In contrast to conventional symbols that correspond to a single nucleotide, a composite symbol represents a mixture of the four nucleotides. Such a symbol is defined by a probability vector ${p_A, p_C, p_G, p_T}$, where $p_\mu$ denotes the proportion of nucleotide $\mu$ and $\sum p_\mu = 1$. As an example, the symbol ${R = (0.1, 0.2, 0.3, 0.4)}$ corresponds to a mixture of all four bases. When a sequence such as $CRG$ is synthesized, the resulting pool contains all four possible sequences at the second position, e.g., $CAG,CCG,CGG$ and $CTG$, appearing with the same frequency as the mixed symbol $R$. By sequencing a subset of molecules from this pool, the underlying nucleotide proportions can be inferred.

A generalization of the composite symbol framework, known as \emph{combinatorial composite DNA}, was introduced by Preuss et al.~\cite{PGYA24}. In this model, composite symbols are defined at the shortmer level, so that each symbol corresponds to a mixture of short DNA sequences, called motifs or building-blocks, rather than to individual nucleotides. Error-correcting codes tailored to the composite DNA setting were studied by Zhang et al. in~\cite{ZC22}. In addition, several related models and associated coding techniques have been proposed in recent works, including~\cite{OM1,OM2}. Other related studies include the selection of composite symbol sets based on decoding probability~\cite{TC24} and capacity results for channels with combinatorial composite symbols~\cite{ROM}. 
It is worth noting that combinatorial composite DNA may be interpreted as a special case of composite symbols with a union-type distribution, or equivalently, as a \emph{subset} of the composite symbol framework.

The idea of \emph{rank modulation} has been extensively investigated in the context of flash memory, for example in~\cite{RAMOSCH,RobRank}. Coding schemes based on permutations were further developed and analyzed in~\cite{CORRSCH,BAMA,RMEitanRyan}.
In this work, we further develop the combination of the theories of composite DNA and rank modulation. While a composite DNA symbol is characterized by its probability distribution, a rank-modulated composite symbol is determined solely by the relative ordering of its motifs. For instance, the symbols $(0.1,0.2,0.3,0.4)$ and $(0.05,0.25,0.3,0.4)$ correspond to distinct composite symbols, yet they define the same rank-modulated symbol since the ordering of their component motifs is identical. Rather than considering the full space of probability distributions, we restrict attention to the \emph{ranks} induced by the symbol distribution.

The work reported in~\cite{RMC25} introduced a channel model based on rank-modulated composite DNA symbols and studied coding over this alphabet. The investigation addressed coding for rank-modulated composite symbols under errors that affect only the relative ordering of motifs/building-blocks. In that setting, the symbols were modeled as \emph{fixed-length} partial permutations, errors were measured using Kendall's tau, and bounds and constructions were derived for detecting and correcting such errors.

 %That work focused on two central problems. It first addresses the capacity of a channel that receives a rank-modulated composite symbol and outputs a single motif according to the symbol’s probability distribution. This channel is used multiple times in every transmission to constitute an alphabet symbol. This problem is closely related to the binomial channel studied in~\cite{WESEL}, as well as to capacity analyses for composite DNA symbols considered in~\cite{KYW23}. 

 %The second problem addressed coding for rank-modulated composite symbols under errors that affect only the relative ordering of motifs. In that setting, the symbols were modeled as \emph{fixed-length} partial permutations, errors were measured using the Kendall's tau distance, and bounds and constructions were derived for detecting and correcting such rank errors.

This work departs from this framework in two fundamental ways. First, we remove the assumption that rank-modulated symbols correspond to permutations of fixed length and instead allow permutations of \emph{variable length}. Second, rather than modeling errors exclusively as changes in the relative order of motifs, we argue that a more faithful model is obtained by considering insertion and deletion errors that occur at the lower end of the ranking, corresponding to weak motifs in the composite symbol. This shift leads to a different coding perspective: we first develop codes over variable-length permutations that can detect and correct insertions and deletions, and only then lift these constructions to the rank-modulated composite setting. In this latter stage, codewords are sequences whose components are variable-length permutations. We provide both constructions and bounds on the size of codes for such alphabets.

The rest of the paper is organized as follows. In Section~\ref{sec:permdef} we define rank modulated codes and some related error models. In this section we also state the connection of different error models over permutations. In Section~\ref{sec:permcon}, we describe code constructions over permutations and prove related bounds. In Section~\ref{sec:defvector}, we extend our investigation to vectors (or equivalently sequences or strings) of partial permutations. These vectors represent words constructed over alphabets of partial permutations.

\section{Definitions, Preliminaries, and Problems Statements}\label{sec:permdef}

\newcommand\partperm[2]{\mathcal{S}_{#1}^{#2}}
\newcommand\partpermgen{\partperm{m}{q}}

\newcommand\partpermall[1]{\partperm{\textbf{all}}{#1}}
\newcommand\partpermallgen{\partperm{\textbf{all}}{q}}

\newcommand\partpermallgenseq{(\partperm{\textbf{all}}{q})^n}

\newcommand{\rankchanneldel}[3]{\textbf{RMCC-DEL}(#1,#2,#3)}
\newcommand{\rankchanneldelconsts}{\rankchanneldel{q}{d}{p}}

\newcommand\insoperatorsphere[1]{\mathcal{S}_\textbf{ins}^{#1}}

\newcommand\deloperator[1]{\mathcal{B}_\textbf{del}^{#1}}
\newcommand\insoperator[1]{\mathcal{B}_\textbf{ins}^{#1}}
\newcommand\suboperator[1]{\mathcal{B}_\textbf{sub}^{#1}}

\newcommand\indeloperator[1]{\mathcal{B}_\textbf{indel}^{#1}}

\newcommand{\codesizedetdel}{\mathsf{DEL}_{\textbf{det}}}
\newcommand{\codesizecordel}{\mathsf{DEL}_{\textbf{cor}}}

\newcommand{\codesizedetdelseq}{\mathsf{DEL}^n_{\textbf{det}}}
\newcommand{\codesizecordelseq}{\mathsf{DEL}^n_{\textbf{cor}}}

\newcommand{\codesizedetins}{\mathsf{INS}_{\textbf{det}}}
\newcommand{\codesizecorins}{\mathsf{INS}_{\textbf{cor}}}

\newcommand{\codesizedetsub}{\mathsf{SUB}_{\textbf{det}}}
\newcommand{\codesizecorsub}{\mathsf{SUB}_{\textbf{cor}}}

\newcommand{\codesizedetindel}{\mathsf{INDEL}_{\textbf{det}}}
\newcommand{\codesizecorindel}{\mathsf{INDEL}_{\textbf{cor}}}

When using DNA for data storage we aim to replace the traditional binary alphabet $\{0,1\}$ to the four nucleotides $\{A,C,G,T\}$. The information is stored as a sequence or a \emph{strand} of nucleotides. 
%For example, a picture can be presented as the sequence $AATGTA$. 
In this scheme, we are limited to $\log_2(4)=2$ bits of information per symbol. In this paper, we use a general alphabet $[q] = \{0, 1, \dots, q-1\}$ of size $q$, which in DNA data storage, $q=4$. %During the synthesis process (writing) many copies of the same strand are generated. %While under the assumption of error free model only one sequencing (reading) is needed to recover the data.

\subsection{Composite Symbols and Combinatorial Composite Symbols}
A recent approach introduced in~\cite{anavy_DataStorageDNA_2019, augmented_encoding} under the name \emph{composite symbols} utilizes the fact that many copies are generated during the synthesis process to increase the number of bits per symbol. Since the synthesis process is much more expensive than sequencing, this idea maintains the same cost of the synthesis process as traditional DNA data storage, while it increases the cost of the sequencing process, mainly, the duration of the process. In \cite{anavy_DataStorageDNA_2019}, the authors presented a new symbol $M$ that can be either $A$ or $C$ during the synthesis process with probabilities $\frac{1}{2},\frac{1}{2}$. For example the sequence $ATTMGA$ can generate during synthesis the following set of strands $\{ATTAGA,ATTAGA,ATTAGA,ATTGCA\}$. One can note that in this method one sequence is not enough to recover the original strand. By reading enough strands one can infer from the fact that in the fourth index there can be either $A$ or $C$ that it is the new symbol $M$. Theoretically, this method allows us to use any probability distribution over $A,C,G,T$ as a new symbol.  In general over the general alphabet $[q]$, any probability distribution in the $(q-1)$-dimensional probability simplex $\Delta_{q} = \left\{ (p_1, p_2, \dots, p_q) \in [0,1]^q : \sum_{i=1}^{q} p_i = 1 \right\}$. Since the set is infinite there is no bound on the number of bits per symbol, however, using the symbols $(\frac{1}{2},\frac{1}{2},0,0)$ and $(\frac{1}{2}-0.001,\frac{1}{2}+0.001,0,0)$ over $A,C,G,T$ for example will require so many reads to distinct between the two that this method will not be useful. 

Previous works \cite{TC24, KYW23} have focused on selecting optimal subsets of composite symbols to minimize the decoding failure error probability or maximize channel capacity. Furthermore, constructions of error-correcting codes specifically designed for composite symbols have been explored in \cite{OM1, OM2}.

%\subsection{Combinatorial Composite Symbols}
To overcome the sensitivity of the reading process, a new approach introduced in \cite{PGYA24} utilizes \emph{combinatorial composite symbols}. Instead of trying to recover the exact probability distribution, the goal is to determine the set of symbols that represent the composite symbol. In this method the probability distribution is always uniform over a constant number of symbols. For a fixed size $m$, the set of combinatorial composite symbols consists of all subsets of $[q]$ with cardinality $m$, i.e.,
$
\Sigma^{q,m}_{\text{comb}} \triangleq \{A \subseteq [q] : |A| = m\}.
$
The size of this expanded alphabet is $\binom{q}{m}$, increasing the number of bits per symbol to $\log_2\left(\binom{q}{m}\right)$. Given $m$ we can stop the reading process assuredly after receiving exactly $m$ symbols in every index. In~\cite{TC24}, the authors calculated the expected number of reads that are needed in this case.

\subsection{Partial Permutations}
A different approach to avoid this sensitivity of the reading process was recently introduced in \cite{RMC25}. To increase even more the number of symbols than in the combinatorial approach, instead of considering all the probability distributions they were interested only in the \emph{ranking} between the symbols of every composite symbol. For example the symbols $(\frac{1}{4},\frac{3}{4},0,0)$ and $(\frac{1}{3},\frac{2}{3},0,0)$ over $A,C,G,T$ represent the same symbol $AC$, i.e., $A$ is less than $C$. In this paper we will also use the notation $A<C$ or $AC$ where the symbols are ordered increasingly from the left to the right.

While in~\cite{RMC25} errors in the ranking were addressed, in this paper we address insertions and deletions in the left tail of the partial permutation.
First we define our base symbols, the \emph{partial permutations}. 
For $m<q$, we denote by $\partpermgen$ the set of all partial permutations over $[q]$ of length $m$. For example, for $q=3$ it holds that
$
\partperm{1}{3}=\{1,2,3\},\partperm{2}{3}=\{12,21,23,32,13,31\},\partperm{2}{3}=\{123,132,213,231,312,321\}.
$
% for $q=3$ and $m=2$ it holds that
% $
% \partperm{2}{3}=\{12,21,23,32,13,31\},
% $
% and for $q=3$ and $m=3$ it holds that
% $
% \partperm{2}{3}=\{123,132,213,231,312,321\}.
% $

In our paper we use partial permutations of \emph{any length},
and denote by $\partpermallgen$ the set of all partial permutations 
over $[q]$ of
any length greater than $1$, i.e.,
% Since we wish in our paper to use partial permutations of \emph{any length}
% we denote by $\partpermallgen$ the set of all partial permutations over $[q]$ of any length greater than $1$, i.e.,
$
\partpermallgen=\cup_{i=1}^{q}{\partperm{i}{q}}.
$
For $q=3$ it holds that
$
\partpermall{3}=\partperm{3}{3}\cup \partperm{2}{3} \cup \partperm{1}{3}.
$
One natural question regarding this set of all partial permutations is finding its size in order to bound the capacity. %This result is given in the following claim.
\begin{claim}
    It holds that
    $
    |\partpermallgen|=q!\sum_{i=0}^{q-1}{\frac{1}{i!}}\approx q!\cdot e.
    $
    % where 
    % \[
    % \lim _{q\rightarrow\infty}{\sum_{i=1}^{q}{\frac{1}{(q-i)!}}}=e
    % \]
\end{claim}

\begin{comment}
\begin{corollary}
    When using partial permutations one can store up to $\log_2(q!\cdot e)$ bits of information per nucleotide.
\end{corollary}
\end{comment}

% \begin{IEEEproof}
%     By definition it holds that
%     \[
%     |\partpermallgen|=\sum_{i=1}^{q}{|\partperm{i}{q}|}.
%     \]
%     There are exactly $\frac{q!}{(q-i)!}$ partial permutations in $\partperm{i}{q}$, therefore, 
%     \begin{align*}
%     |\partpermallgen|=\sum_{i=1}^{q}{\frac{q!}{(q-i)!}}=q!\sum_{i=1}^{q}{\frac{1}{(q-i)!}} \approx q!\cdot e.
%     \end{align*}
%     where 
%     \[
% {\sum_{i=0}^{\infty}{\frac{1}{i!}}}=e
%     \]
% \end{IEEEproof}

% During synthesis(writing) we can only use real DNA symbols, i.e., $A,C,G$ and $T$. We can use partial permutations in the same way we use \emph{composite symbols}. We choose a probability distribution that represents our complex symbol. The left side of the partial permutation represents the weak part, i.e., the symbols with smaller probability. We also use the terms suffix and prefix for the right and left tail of the partial permutation, respectively.

\begin{example}\label{ex:compositeexample}
    The partial permutation $AC$ represents the case where $A<C$. By using $A,C$ with the probabilities $1/3,2/3$ we store the information of the ranking. In the sequencing process we count the number of occurrences of each symbol and infer the original partial permutation. If we get $A,C$ for $10,17$ times we assume that the partial permutation was $AC$, i.e. $A<C$. However, since it is a probabilistic process it is not guaranteed that the output in the sequencing process will represent the original partial permutation. It is possible to calculate that by sequencing $10$ times the probabilities to get $A<C$ and $C\leq A$ are $0.787$ and $0.213$ respectively.
    \end{example}

%After defining the partial permutations, we are interested in related codes. 
In this paper we call a set of partial permutations $\mathcal{C}\subseteq \partpermallgen$ a \emph{code over permutations} or simply \emph{a code}. When using partial permutations as DNA symbols, natural errors arise like deletions, insertions, and deletions combined with insertions. In this paper, all those errors will occur in the \emph{left tail} of the permutation, i.e., the left side of the permutation. Such errors are frequent when using DNA for storage. 

% \tomer{The next example and the table are not good. I've tried a few variations but they all fail to capture the crux of this model.}
\begin{example}
During the synthesis of the same partial permutation from Example~\ref{ex:compositeexample}, i.e., $A<C$ with probabilities $1/3,2/3$ it is common to have errors. For example, it is possible that with probability $1/100$ we will write $T$ instead of $A$ or $C$. Since we intended to write a partial permutation of size $2$ we only care about the $2$ most frequent symbols in the output and their ranking. In Table~\ref{tab:dna_probs_desc} we have the $5$ strongest outputs for the case of $10$ samples, probabilities $1/3,2/3$ for the partial permutation and error probability of $1/100$. The distribution of the symbols $A,C,T$ is multinomial with probabilities $0.33,0.66,0.01$.
% With error probability $\epsilon$ the distribution of the symbols $A,C,T$ is multinomial with probabilities $\epsilon,\frac{1-\epsilon}{3},\frac{2(1-\epsilon)}{3}$  
\end{example}
\begin{table}[ht]
    \centering
    \caption{Probabilities of Symbol Order Outcomes}
    \label{tab:dna_probs_desc}
    \begin{tabular}{lcl}
        \toprule
        \textbf{Order Outcome} & \textbf{Probability} & \textbf{Description} \\
        \midrule
        $A < C$     & 0.695949 & Correct \\
        $C < A$     & 0.069227 & Swap \\
        $T < A < C$ & 0.066184 & $1$ Tail Insertion  \\
        $C$         & 0.015683 & $1$ Tail Deletion \\
        $T < C < A$ & 0.013427 & $1$ Tail Insertion + Swap \\
        \bottomrule
    \end{tabular}
\end{table}

One can note that the probabilities of errors in the tail of the symbol are just as likely as the swap. This gap is addressed in our paper to detect and correct such errors in the tail. The exact error models are introduced in the following subsection.
\subsection{Error Models}\label{subsec:models}
First, we address the problem of deletions in the left tail.
For $\pi=(\pi_1,\pi_2,\dots,\pi_m)\in\partpermgen$, we say that $t\leq m-1$ \emph{tail deletions} in $\pi$ occurred if the first $t$ symbols were deleted. It is important to note that at most $m-1$ deletions can occur on a partial permutation from $\partpermgen$.

For $\pi=(\pi_1,\pi_2,\dots,\pi_i)\in\partpermallgen$, we denote by $\pi_{\downarrow_{j}}$ the partial permutation $\pi$ after $j$ tail deletions. For example,
$
(2341)_{\downarrow_2}=41.
$
Since we cannot delete the last symbol it holds that
$
(41)_{\downarrow_2}=1.
$
We also define $\pi_{\downarrow_{0}}=\pi.$

We denote by $\deloperator{t}(\pi)$ the set of all partial permutations after \emph{at most} $t$ tail deletions. For example, if $2$ tail deletions occurred in $3245\in\partperm{4}{6}$ we get the partial permutation $45$, and furthermore, $\deloperator{2}(3245)=\{3245,245,45\}$. One can observe that $\deloperator{t}(\pi)=\{\pi,\pi_{\downarrow_1},\pi_{\downarrow_2},\dots,\pi_{\downarrow_{t}}\}$.

% Next we define the tail deletions channel:
% Let $q$ be the size of the alphabet.
% We denote the channel by $\rankchanneldelconsts$.

% \textbf{Input:}
% A partial permutation $\pi=(\pi_1,\pi_2,\dots,\pi_m)\in \partpermall{q}$ 

% \textbf{Output:}
% The channel transmits a partial permutation $\pi'\in\partpermall{q}$.
% With probability $(1-p)$ it holds that $\pi'=\pi$.
% With probability $p$ at most $d$ tail deletions occurred to $\pi$.
% Denote by $\mathsf{cap}_\textbf{del}(q,d,p)$ the capacity of the channel $\rankchanneldelconsts$

A code is called a \emph{$t$-tail-deletion-detecting code} if for every $\pi\in \mathcal{C}$ it holds that
$
\mathcal{C}\cap \deloperator{t}(\pi)= \{\pi\}.
$
We denote by $\codesizedetdel(q,t)$ the size of the largest $t$-tail-deletion-detecting code. A code $\mathcal{C}\subseteq \partpermallgen$ is called a \emph{$t$-tail-deletion-correcting code} if for every $\pi_1,\pi_2\in \mathcal{C}$ such that $\pi_1\neq \pi_2$ it holds that
$
\deloperator{t}(\pi_1) \bigcap \deloperator{t}(\pi_2)= \phi.
$
We denote by $\codesizecordel(q,t)$ the size of the largest $t$-tail-deletion-correcting code.

Secondly, we address similarly the problem of insertions in the tail.
For $\pi=(\pi_1,\pi_2,\dots,\pi_m)\in\partpermgen$, we say that \emph{$t$ tail-insertions} occurred in $\pi$ if $t$ symbols were added to the left tail. Since we are interested in the ranking between the symbols, only new symbols can be inserted, and therefore, at most $q-m$ insertions can occur on a partial permutation from $\partpermgen$. We denote by $\insoperator{t}(\pi)$ the set of all partial permutations after \emph{at most} $t$ tail insertions.
For example, if $1$ tail insertion occurred in $3245\in\partperm{4}{6}$, we can get the partial permutations $13245$ and $63245$. Furthermore, $\insoperator{1}(3245)=\{3245,13245,63245\}.$
Since we only care about the ranking between the symbols, the largest partial permutation can be of size $q$. For example, with $q=4$ it holds that,
$
\insoperator{2}(341)=\{341,2341\}.
$

For $\pi=(\pi_1,\pi_2,\dots,\pi_i)\in\partpermallgen$, we denote by $\insoperatorsphere{t}(\pi)$
the \emph{set} of all partial permutations $\pi'$ such that after \emph{exactly} $t$-tail deletions from $\pi'$ we get $\pi$, i.e.,
$
\pi'_{\downarrow_t}=\pi.
$
For example, 
$
\insoperatorsphere{2}(41)=\{2341,3241\}.
$
It is important to note that the definition of $\insoperatorsphere{t}$ depends on the size of the alphabet. When we use this notation in this paper, the size of $q$ will be clear. One can observe that $\insoperator{t}(\pi)=\cup_{i\leq t}\insoperatorsphere{t}(\pi)$, or similarly $\insoperatorsphere{t}(\pi)=\insoperator{t}(\pi) \setminus \insoperator{t-1}(\pi)$. 

If $\pi\in\partperm{i}{q}$ and $t\leq q-i$ then it holds that
$
|{\insoperatorsphere{t}(\pi)}|=\binom{q-i}{t} t!.
$
For example, with $q=4$ it holds that 
$\insoperatorsphere{2}(1)=\{231,241,321,341,421,431\}$, and $|\insoperatorsphere{2}(1)|=\binom{4-1}{2} 2!=6$.

Similarly to the deletion model, we define a \emph{$t$-tail-insertion-detecting code} and a \emph{$t$-tail-insertion-correcting code}. We denote by $\codesizedetins(q,t)$ and $\codesizecorins(q,t)$ the sizes of the largest $t$-tail insertion-detecting code and correcting code, respectively.

% A code $\mathcal{C}\subseteq \partpermallgen$ is called a \emph{$t$-tail-insertion-detecting code} if for every $\pi\in \cC$ it holds that
% \[
% \cC\cap \insoperator{t}(\pi)= \{\pi\}.
% \]
% We denote by $\codesizedetins(q,t)$ the size of the largest $t$-tail-insertion-detecting code.
% A code $\mathcal{C}\subseteq \partpermallgen$ is called a \emph{$t$-tail-insertion-correcting code} if for every $\pi_1,\pi_2\in \cC$ such that $\pi_1\neq \pi_2$ it holds that
% $
% \insoperator{t}(\pi_1) \bigcap \insoperator{t}(\pi_2)= \phi.
% $
% We denote by $\codesizecorins(q,t)$ the size of the largest $t$-tail-insertion-correcting code.

Lastly, we address the problem of combining insertions and deletions in the tail.
Let $\pi=(\pi_1,\pi_2,\dots,\pi_m)\in\partpermgen$, we say that \emph{$t$ tail indels} occurred if $t$ symbols were deleted or inserted in the left tail. We denote by $\indeloperator{t}(\pi)$ the set of all partial permutations after \emph{at most} $t$ tail indels. It can be verified  that 
\[
\indeloperator{t}(\pi)=\cup _{i=0}^{t}\left\{\cup_{\pi '\in {\deloperator{i}({\pi})}}\left\{\insoperator{t-i}({\pi '})\right\}\right\}.
\]
For example, if $1$ tail indel occurred in $3245\in\partperm{4}{6}$ we can get the partial permutations $245,13245,63245$, and furthermore, $\indeloperator{1}(3245)=\{245,13245,63245\}.$
One should note that the ranking can include all the original symbols but in a different order; for example,
$
2345\in \indeloperator{4}(3245).
$
Furthermore, lower symbols can remain the same while other change. For example,
$
3645\in \indeloperator{4}(3245),
$
with the process of $3245\rightarrow_{d}245\rightarrow_{d}45\rightarrow_{i}645\rightarrow_{i}3645$.

Similarly to the deletion model we define a \emph{$t$-tail-indel-detecting code} and a \emph{$t$-tail-indel-correcting code}. We denote by $\codesizedetindel(q,t)$ and $\codesizecorindel(q,t)$ the sizes of the largest $t$-tail indel-detecting code and correcting code, respectively.

\subsection{Problem Statement}

% \begin{problem}
% For every $q,d,p$ find $\mathsf{cap}_\textbf{del}(q,d,p)$. Furthermore, find the c.a.i.d of the channel $\rankchanneldelconsts$
% \end{problem}

The main problem of this paper is to find codes over permutations that can detect and correct the different types of error models.
\begin{problem}
For every $q,t$ such that $t\leq q$:

\begin{itemize}
    \item Find $\codesizedetdel(q,t),\codesizecordel(q,t)$, and find such codes.
    \item Find $\codesizedetins(q,t),\codesizecorins(q,t)$, and find such codes.
    \item Find $\codesizedetindel(q,t),\codesizecorindel(q,t)$, and find such codes.
\end{itemize}

\end{problem}

% \begin{problem}
% For every $q,t$ find $\codesizedetins(q,t),\codesizecorins(q,t)$, and find such codes.
% \end{problem}

% \begin{problem}
% For every $q,t$ find $\codesizedetindel(q,t),\codesizecorindel(q,t)$, and find such codes.
% \end{problem}

Since every $t$-tail-indel-detecting code is also a $t$-tail-insertion-detecting code and a $t$-tail-deletion-detecting code it holds that 
$
\codesizedetindel(q,t)\leq \codesizedetdel(q,t),\codesizecorindel(q,t)\leq \codesizecordel(q,t)
$ and $
\codesizedetindel(q,t)\leq \codesizedetins(q,t),\codesizecorindel(q,t)\leq \codesizecorins(q,t)
$.

In our paper we mainly focus on tail deletion codes. In Section~\ref{sec:equiverrors}, we present the equivalence and relations between the different error models. The construction of such codes will also bound the size of codes that the can address other error models.

In Section~\ref{sec:defvector}, we extend our investigation to sequences over partial permutation symbols and formally state the related problem.

\section{Equivalence and Relations between Different Error Models}\label{sec:equiverrors}

The equivalence between deletion, insertion and indels codes over vectors is well known~\cite{Levenshtein1966}. In our model, the codewords are of different lengths and therefore, behave differently.
The first theorem presents the equivalence between tail deletion and insertion detecting codes and disproves an equivalence to tail indel detecting codes. Moreover, for $q\geq 3$ and for any $t\geq 1$ being able to detect $t$ deletions cannot imply being able to detect even a constant amount of indels

\begin{theorem}\label{th:allsamecor}

A code $\mathcal{C}$ is a $t$-tail-deletion-detecting code if and only if it is 
a $t$-tail-insertion-detecting code.
However, for any $q\geq 3$ and $t\geq 1$ there exists a code $\mathcal{C}$ that is a $t$-tail-deletion-detecting code that is not a $2$-tail-indel-detecting code.

% The following statements are equivalent 
% \begin{enumerate}
% \item A code over permutations $\mathcal{C}\subseteq \partpermallgen$ is a $t$-tail-deletion-detecting code.
% \item A code over permutations $\mathcal{C}\subseteq \partpermallgen$ is a $t$-tail-insertion-detecting code.
% % \item A code over permutations $\mathcal{C}\subseteq \partpermallgen$ is a $t$-tail-indel-detecting code.
% \end{enumerate}

% Let $q\geq2$, for any $t\geq 1$ exists a code over permutations $\mathcal{C}\subseteq \partpermallgen$ that is a $t$-tail-deletion-detecting code that is not a $2$-tail-indel-detecting code.

% Furthermore,
% There exists a code over permutations $\mathcal{C}\subseteq \partpermallgen$ that is a $t$-tail-deletion-detecting code that is not a $t$-tail-indel-detecting code.
\end{theorem}

\begin{IEEEproof}
    Let $\mathcal{C}$ be a $t$-tail-deletion-detecting-code. Assume towards contradiction that $\cC$ is not a a $t$-tail-insertion-detecting code, i.e., there exist $x,y\in \cC$ such that 
    $
    y\in \insoperator{t}(x).
    $
    There exists $\omega\in \partpermallgen$ such that $|\omega|\leq t$ and $y=\omega x$, i.e.,
    $
    x\in \deloperator{t}(y),
    $
    which results with a contradiction.
    Let $\mathcal{C}$ be a $t$-tail-insertion-detecting-code. Assume towards contradiction that $\cC$ is not a a $t$-tail-deletion-detecting code, i.e., there exist $x,y\in \cC$ such that 
    $
    y\in \deloperator{t}(x).
    $
    There exists $\omega\in \partpermallgen$ such that $|\omega|\leq t$ and $x=\omega y$, i.e.,
    $
    x\in \insoperator{t}(y),
    $
    which results with a contradiction.

    For the last part of the theorem, for any $q\geq 3$ and $t\geq 1$ the code $\cC=\{12,32\}$ is $t$-tail-deletion-detecting code, however, it is not a $2$-tail-indel-detecting code since $12\in \indeloperator{2}(32)$.
    \end{IEEEproof}

% The equivalence does not include tail-indel-detecting codes as shown in the following theorems
% \begin{theorem}\label{th:deldetnotindeldet}
%     There exists a code over permutations $\mathcal{C}\subseteq \partpermallgen$ that is a $t$-tail-deletion-detecting code that is not a $t$-tail-indel-detecting code.
% \end{theorem}
% \begin{IEEEproof}
%     For example, if $q=4$ the code $\cC=\{12,32\}$ is $2$-tail-deletion-detecting code, however, it is \textbf{not} a $2$-tail-indel-detecting code since $12\in \indeloperator{2}(32)$.
%     \end{IEEEproof}
    
% Moreover, for $q\geq 2$ and for any $t$ being able to detect $t$ deletions cannot imply being able to detect even a constant amount of indels such as $2$ indels as shown in the following theorem,
% \begin{theorem}
%     Let $q\geq2$, for any $t\geq 1$ exists a code over permutations $\mathcal{C}\subseteq \partpermallgen$ that is a $t$-tail-deletion-detecting code that is not a $2$-tail-indel-detecting code.
% \end{theorem}
% \begin{IEEEproof}
%     For example, the same code from the proof of Theorem~\ref{th:deldetnotindeldet}, i.e., $\cC=\{12,32\}$ is $t$-tail-deletion-detecting code for any $t$, however, it is \textbf{not} a $2$-tail-indel-detecting code.
%     \end{IEEEproof}

Lastly, we show in the following theorem that being able to detect $t$ tail insertions implies being able to correct $t$ tail insertions.
\begin{theorem}\label{th:allsamedel}
    If a code $\mathcal{C}$ is a $t$-tail-insertion-detecting code then it is a $t$-tail-insertion-correcting code.
\end{theorem}
\begin{IEEEproof}
    Let $\mathcal{C}$ be a $t$-tail-insertion-detecting code. Assume towards contradiction that $\cC$ is not a a $t$-tail-insertion-correcting code, i.e., there exist distinct $x,y\in \cC,z\in \partpermallgen$ such that 
    $
    z\in \insoperator{t}(x)\cap\insoperator{t}(y).
    $
    In other words there exist $x',y'$ such that $x'x=z=y'y$. If $|x|=|y|$ then $x=y$, which results with a contradiction. Assume w.l.o.g that $|y|>|x|$, one can note that $x$ is a suffix of $y$ and since $|z|-|x|\leq t$ then $|y|-|x|\leq t$, hence, 
    $
    y\in \insoperator{t}(x),
    $
    which results with a contradiction.
\end{IEEEproof}

\begin{corollary}
It holds that 
$
\codesizedetindel(q,t)\leq \codesizedetdel(q,t)=\codesizedetins(q,t)=\codesizecorins(q,t).
$
\end{corollary}

% \begin{corollary}
% It holds that 
% \[
% \codesizedetins(q,t)=\codesizecorins(q,t).
% \]
% \end{corollary}
% \begin{theorem}\label{th:allsamedel}
% The following statements are equivalent 
% \begin{enumerate}
% \item A code over permutations $\mathcal{C}\subseteq \partpermallgen$ is a $2t$-tail-deletion-correcting code.
% \item A code over permutations $\mathcal{C}\subseteq \partpermallgen$ is a $t$-tail-indel-correcting code.
% \end{enumerate}
% \end{theorem}

Next we show some relations between the different types of tail correcting codes. First, we show in the following theorems the relation between tail indel and deletion correcting codes.

\begin{theorem}\label{th:allsamedel}
    Let $t\geq 1$. If a code $\mathcal{C}$ is a $t$-tail-deletion-correcting code then it is a $t$-tail-insertion-correcting code.
However, the opposite does not hold.
    
    \end{theorem}
\begin{IEEEproof}
    Let $\mathcal{C}$ be a $t$-tail-deletion-correcting code. Assume towards contradiction that $\cC$ is not a $t$-tail-insertion-correcting code, i.e., there exist distinct $x,y\in \cC,z\in \partpermallgen$ such that 
    $
    z\in \insoperator{t}(x)\cap\insoperator{t}(y).
    $
    In other words there exist $x',y'\in \partpermallgen$ such that $x'x=z=y'y$. If $|x|=|y|$ then $x=y$, which results with a contradiction. Assume w.l.o.g that $|y|>|x|$, one can note that there exists $w\in\partpermallgen$ such that $wx=y$.
    From $|z|-|x|\leq t$ it holds that $|y|-|x|\leq t$ since
    $t\geq |z|-|x|=|z|-|y|+|y|-|x|\geq |y|-|x|$.
    Hence, 
    $
    x\in \deloperator{t}(y),
    $
    which results with a contradiction.

    For the second part of the theorem, with $q\geq4$ and $t\geq 1$ the code $\cC=\{23,43\}$ is $t$-tail-insertion-correcting code, however, it is not a $t$-tail-deletion-correcting code since $3\in \deloperator{1}(23)\cap\deloperator{1}(43)$.
\end{IEEEproof}

% It is important to note that the opposite does not hold.
% \begin{theorem}\label{th:deltoins}
%     There exists a code over permutations $\mathcal{C}\subseteq \partpermallgen$ that is a $t$-tail-insertion-correcting code that is not a $t$-tail-deletion-correcting code.
% \end{theorem}
% \begin{IEEEproof}
%     For example, if $q=4$ the code $\cC=\{23,43\}$ is $1$-tail-insertion-correcting code, however, it is \textbf{not} a $t$-tail-deletion-correcting code since $3\in \deloperator{1}(23)\cap\deloperator{1}(43)$.
    
%     \end{IEEEproof}

\begin{corollary}
It holds that 
$
\codesizecordel(q,t)\leq \codesizecorins(q,t).
$
\end{corollary}

Next, we show that correcting $2t$ deletions implies correcting $t$ indels. Moreover, the use of $2t$ is necessary. However, the opposite does not hold.

\begin{theorem}\label{th:allsamedel}
It holds that
\begin{enumerate} 
    \item If a code $\mathcal{C}$ is a $2t$-tail-deletion-correcting code then it is a $t$-tail-indel-correcting code. However, the opposite does not hold.

    \item For every $t\geq 1$ there exists a code $\mathcal{C}$ that is a $(2t-1)$-tail-deletion-correcting code that is not a $t$-tail-indel-correcting code.

    % \item There exists a code over permutations $\mathcal{C}\subseteq \partpermallgen$ that is a $t$-tail-indel-correcting code that is not a $2t$-tail-deletion-correcting code.
\end{enumerate}

    % Moreover, the use of $2t$ is necessary.

    % However, the opposite does not hold. 
    
    \end{theorem}
\begin{IEEEproof}

$(1)$: Let $\mathcal{C}$ be a $2t$-tail-deletion-correcting code. Assume towards contradiction that $\cC$ is not a $t$-tail-indel-correcting code, i.e., there exist $x_1,x_2\in \cC,z\in \partpermallgen$ such that 
    $
    z\in \indeloperator{t}(x_1)\cap\indeloperator{t}(x_2).
    $
    Since only $t$ operations occurred it holds that $||x_i|-|z||\leq t$ for $i\in \{1,2\}$.
    Let $x_i'$ be the longest mutual suffix of $x_i$ and $z$ for $i\in \{1,2\}$, and assume w.l.o.g that $|x_1'|\leq |x_2'|$. 
    Since $z\in \indeloperator{t}(x_i)$ and $x_i'$ is the longest mutual suffix of $x_i$ and $z$ then $z,x_i\in \insoperator{t}(x'_i)$ and $|z|-|x_i'|\leq t,|x_i|-|x_i'|\leq t$ for $i\in \{1,2\}$. 
    If $z\in \insoperator{t}(x_1')\cap \insoperator{t}(x_2')$ and $|x_1'|\leq |x_2'|$ then $x_1'$ is a suffix of $x_2'$, furthermore, $|x_2'|-|x_1'|\leq t$. Hence, $x_1'\in \deloperator{t}(x_2')$.
    If  $x_i\in \insoperator{t}(x_i')$ then $x_i'\in \deloperator{t}(x_i)$ for $i\in \{1,2\}$.
    If $x_1'\in \deloperator{t}(x_2')$ then $x_1'\in \deloperator{2t}(x_2)$. Since $x_1'\in \deloperator{t}(x_1)$ then $x_1'\in \deloperator{2t}(x_1)$, therefore, $x_1'\in \deloperator{2t}(x_1)\cap \deloperator{2t}(x_2)$,  which results with a contradiction.

   For the second part of the theorem, for example, for $q=4$ the code $\cC=\{321,241\}$ is $1$-tail-indel-correcting code, however, it is not a $2$-tail-deletion-correcting code since $1\in \deloperator{1}(321)\cap\deloperator{1}(241)$.

    $(2)$:
Let $t\geq 1$, the code $\cC=\{1,123\dots (2t+1)\}$ is a $2t-1$-tail-deletion-correcting code, however, it is not a $t$-tail-indel-correcting code since $(t+1)\dots 321\in \indeloperator{t}(1)\cap\indeloperator{t}(1, (2t+1)\dots 321).$
    
\end{IEEEproof}

% It is important to note that the opposite does not hold.
% \begin{theorem}\label{th:deltoins}
%     There exists a code over permutations $\mathcal{C}\subseteq \partpermallgen$ that is a $t$-tail-indel-correcting code that is not a $2t$-tail-deletion-correcting code.
% \end{theorem}
% \begin{IEEEproof}
%     For example, if $q=4$ the code $\cC=\{321,241\}$ is $1$-tail-indel-correcting code, however, it is \textbf{not} a $2$-tail-deletion-correcting code since $1\in \deloperator{1}(321)\cap\deloperator{1}(241)$.
    
%     \end{IEEEproof}

% Moreover, the use of $2t$ is necessary as shown in the following theorem,
% \begin{theorem}\label{th:allsamedel}
%     For every $t\geq 1$ there exists a code over permutations $\mathcal{C}\subseteq \partpermallgen$ that is a $2t-1$-tail-deletion-correcting code that is not a $t$-indel-deletion-correcting code.
% \end{theorem}
% \begin{IEEEproof}
%     For example, the code $\cC=\{1,123\dots (2t+1)\}$ is $2t-1$-tail-deletion-correcting code, however, it is \textbf{not} a $t$-tail-indel-correcting code since $$123\dots(t+1)\in \indeloperator{t}(1)\cap\indeloperator{t}(1,123\dots (2t+1)).$$
    
%     \end{IEEEproof}

\begin{corollary}
It holds that 
% \[
% \codesizecordel(q,2t)\leq \codesizecorindel(q,t).
% \]
% Since every $t$-tail-indel-correcting code is also a $t$-tail-deletion-correcting code it holds that 
$
\codesizecordel(q,2t)\leq \codesizecorindel(q,t) \leq \codesizecordel(q,t) .
$
\end{corollary}

\section{Code Constructions over Permutations}\label{sec:permcon}

In this section, we present constructions of $t$-tail-deletion detecting and correcting codes. 

% will define some notations that will be used in the rest of this paper.

% Let $\pi=(\pi_1,\pi_2,\dots,\pi_i)\in\partpermallgen$, we denote by $\pi\downarrow_j$ the partial permutation $\pi$ after $j$-tail deletions. For example 
% $
% (2341)\hspace{-1ex}\downarrow_2=41.
% $
% Since we cannot delete the last symbol it holds that,
% $
% (41)\hspace{-1ex}\downarrow_2=1.
% $
% Furthermore, one can observe that $\deloperator{t}(\pi)=\{\pi,\pi\downarrow_1,\pi\downarrow_2,\dots,\pi\downarrow_t\}$.
% We also define $\pi \downarrow_0=\pi.$

% For the next theorem we need the following construction.

\begin{theorem}\label{th:bestdetperm}
 Let $q$ be the size of the alphabet, and let $t< q$. Denote,
% $$
% \cC_{q,t}^{\mathsf{det}}=\bigcup_{1\leq i\leq q : \ i\equiv q\bmod (t+1)}\partperm{i}{q}.
% $$
$
\cC_{q,t}^{\mathsf{det}}=\bigcup_{i=0}^{\left\lfloor\frac{q-1}{t+1}\right\rfloor}\partperm{q-i(t+1)}{q}.
$
It holds that $\cC_{q,t}^{\mathsf{det}}$ is a $t$-tail-deletion-detecting code. Furthermore, $\cC_{q,t}^{\mathsf{det}}$ is an optimal code and so 
% $\codesizedetdel(q,t)=|\cC_{q,t}^{\mathsf{det}}|=\sum_{{1\leq i\leq q : \  i\equiv q\bmod (t+1)}}{|\partperm{i}{q}|}.$
$\codesizedetdel(q,t)=|\cC_{q,t}^{\mathsf{det}}|=\sum_{i=0}^{\left\lfloor\frac{q-1}{t+1}\right\rfloor}{|\partperm{q-i(t+1)}{q}|}=q!\cdot\sum_{i=0}^{\left\lfloor\frac{q-1}{t+1}\right\rfloor}{{\frac{1}{(i(t+1))!}}}.$
\end{theorem}

\begin{IEEEproof}
First, we will show that $\cC_{q,t}^{\mathsf{det}}$ is a $t$-tail-deletion-detecting code.
Let $\pi=(\pi_1,\pi_2,\dots,\pi_j)\in\cC_{q,t}^{\mathsf{det}}$ where $j= q-i(t+1)$ for $i$ such that $0\leq i\leq {\left\lfloor\frac{q-1}{t+1}\right\rfloor}$. It holds that
% \[
% \deloperator{t}(\pi)=\{(\pi_1,\pi_2,\dots,\pi_i),(\pi_2,\pi_3,\dots\pi_i),\dots,(\pi_t,\pi_{t+1},\dots\pi_i)\},
% \]
$\deloperator{t}(\pi)=\{\pi,\pi_{\downarrow_1},\pi_{\downarrow_2},\dots,\pi_{\downarrow_{t}}\}$, therefore, for every $\ell\in\{j-t,j-t+1,\dots,j-1\}\cap \mathbf{N}$ it holds that
$\ell \not= q-i(t+1)$ for $i$ such that $0\leq i\leq {\left\lfloor\frac{q-1}{t+1}\right\rfloor}$, and thus,
$
\mathcal{C}\cap \deloperator{t}(\pi)= \{\pi\}.
$

Now we will show that $\cC_{q,t}^{\mathsf{det}}$ is optimal. First, we find an upper bound for $t$-tail-deletion-detecting codes. Let $\cC$ be a $t$-tail-deletion-detecting code. Let $\pi=(\pi_1,\pi_2,\dots,\pi_j)\in\partperm{j}{q}$
where $j= q-i(t+1)$ for $i$ such that $0\leq i\leq {\left\lfloor\frac{q-1}{t+1}\right\rfloor}$.
It holds that at most one of the permutations in the set $\{\pi,\pi_{\downarrow_1},\pi_{\downarrow_2},\dots,\pi_{\downarrow_{t}}\}$ can be in $\cC$ since it is a $t$-tail-deletion-detecting code, and for every $0\leq\ell<k\leq t$ it holds that $(\pi_{\downarrow_\ell})_{\downarrow_{(k-\ell)}}=\pi_{\downarrow_k}$. Since $0<k-\ell\leq t$, it holds that $\pi_{\downarrow_{k}}\in \deloperator{t}(\pi_{\downarrow_\ell})$. 
It holds that 
$
\bigcup_{i=0}^{\left\lfloor\frac{q-1}{t+1}\right\rfloor}\bigcup_{\pi\in \partperm{q-i(t+1)}{q}}{\deloperator{t}(\pi)}=\partpermallgen.
$
 % One can note that the sets  $\{\pi,\pi_{\downarrow_1},\pi_{\downarrow_2},\dots,\pi_{\downarrow_{t}}\}$ for every $\pi\in\partperm{j}{q}$ for any
 % $j=q-i(t+1)$ such that $0\leq i\leq {\left\lfloor\frac{q-1}{t+1}\right\rfloor}$
 % cover $\partpermallgen$,
Hence,
$
|\cC|\leq \sum_{i=0}^{\left\lfloor\frac{q-1}{t+1}\right\rfloor}{|\partperm{q-i(t+1)}{q}|}.
$ It is important to note that the statement holds even though the deletion balls are not necessarily disjoint.

By construction it holds that $
\cC_{q,t}^{\mathsf{det}}=\bigcup_{i=0}^{\left\lfloor\frac{q-1}{t+1}\right\rfloor}\partperm{q-i(t+1)}{q}.
$
% for any $i\equiv q\bmod (t+1)$ it holds that $\partperm{i}{q}\subset\cC_{q,t}^{\mathsf{det}}$. 
Therefore,
$
|\cC_{q,t}^{\mathsf{det}}|= \sum_{i=0}^{\left\lfloor\frac{q-1}{t+1}\right\rfloor}{|\partperm{q-i(t+1)}{q}|}
$,
and hence, the code $\cC_{q,t}^{\mathsf{det}}$ is optimal.

% contains $\pi$, therefore, it contains exactly one element in $\{\pi,\pi\downarrow_1,\pi\downarrow_2,\dots,\pi\downarrow_t\}$. One can note that the sets  $\{\pi,\pi\downarrow_1,\pi\downarrow_2,\dots,\pi\downarrow_t\}$ for $\pi\in\partperm{i}{q}$ for any $i\equiv q\bmod (t+1)$ are disjoint and cover $\partpermallgen$,
% hence,
% $
% |\cC|\leq |\cC_{q,t}|.
% $

% Let $\cC\subseteq\partpermallgen$ be a largest $t$-tail-deletion-detecting code different than $\cC_{q,t}^{\mathsf{det}}$. Let Let $\pi=(\pi_1,\pi_2,\dots,\pi_i)\in\cC\setminus\cC_{q,t}^{\mathsf{det}}$ be a partial permutation in $\cC$ and not $\cC_{q,t}^{\mathsf{det}}$ of longest length. It holds that $i<q$ since $\partperm{q}{q}\subset\cC_{q,t}^{\mathsf{det}}$. 
% Furthermore, it holds that $i\not\equiv q\bmod (t+1)$. Let $i<j\leq q$ be the smallest integer such that $i\not\equiv q\bmod (t+1)$.

\end{IEEEproof}

For the rest of the paper, we order the partial permutations in $\insoperatorsphere{t}(\pi)$ in lexicographic increasing order and now we can use the $i$'th partial permutation with the notation of $(\insoperatorsphere{t}(\pi))_i$.
For example, with $q=4$ it holds that $
\insoperatorsphere{2}(41)=\{2341,3241\}
$. Therefore,
${(\insoperatorsphere{2}(41))_1}=2341$ and ${(\insoperatorsphere{2}(41))_2}=3241.
$

For $\cC\subseteq \partpermallgen$, we denote by $(\insoperatorsphere{t}(\cC))_j$ the union of all $(\insoperatorsphere{t}(\pi))_j$ for $\pi\in \cC$. For example, for $q=4$ if $\cC=\{41,1\}$ it holds that $(\insoperatorsphere{2}(\cC))_1=\{2341,231\}$ and $(\insoperatorsphere{2}(\cC))_2=\{3241,241\}$.
The following code is our starting point for the construction of $t$-tail-deletion-correcting codes. Denote,
$
\cC_{q,t}^{\mathsf{base}}=\bigcup_{i=0}^{\left\lfloor\frac{q-t-1}{t+1}\right\rfloor}
\partperm{q-t-i(t+1)}{q}.
$

% $$
% \cC_{q,t}^{\mathsf{base}}=\bigcup_{1\leq i\leq q-t : \ i\equiv q-t\bmod (t+1)}\partperm{i}{q}.
% $$

\begin{lemma}\label{lm:ballinup}
Let $q$ be the size of the alphabet, and let $t< q$. For every $1\leq j\leq t!$ it holds that the code $\left(\insoperatorsphere{t}(\cC_{q,t}^{\mathsf{base}})\right)_j$ is a $t$-tail deletion-correcting code and $
|\left(\insoperatorsphere{t}(\cC_{q,t}^{\mathsf{base}})\right)_{j}|= |\cC_{q,t}^{\mathsf{base}}|
=q!\cdot\sum_{i=0}^{\left\lfloor\frac{q-t-1}{t+1}\right\rfloor}{{\frac{1}{(t+i(t+1))!}}}.$ Furthermore, for every $1\leq j_1< j_2\leq t!$ it holds that the codes
$
\left(\insoperatorsphere{t}(\cC_{q,t}^{\mathsf{base}})\right)_{j_1},\left(\insoperatorsphere{t}(\cC_{q,t}^{\mathsf{base}})\right)_{j_2}$
are mutually disjoint.
Furthermore, if $q\not\equiv 0\bmod (t+1)$ it holds that the code $\left(\insoperatorsphere{t}(\cC_{q,t}^{\mathsf{base}})\right)_{1}\cup \partperm{1}{q}$ 
 is a $t$-tail deletion-correcting code.
\end{lemma}
\begin{IEEEproof}
Let $1\leq j\leq t!$ and distinct $\pi_1,\pi_2\in \left(\insoperatorsphere{t}(\cC_{q,t}^{\mathsf{base}})\right)_j$.
By construction, there exist distinct $x_1,x_2\in \cC_{q,t}^{\mathsf{base}}$ and $y_1,y_2\in \partperm{t}{q}$ such that $\pi_i=y_ix_i$ for $i\in \{1,2\}$. It holds that $\deloperator{t}(\pi_i)=\{{y_i}_{\downarrow_0} x_i,{y_i}_{\downarrow_1}x_i,\dots,{y_i}_{\downarrow_t}x_i\}$. Since $x_1\neq x_2$ it holds that $\deloperator{t}(\pi_1)\cap \deloperator{t}(\pi_2)=\phi$, hence, the code $\left(\insoperatorsphere{t}(\cC_{q,t}^{\mathsf{base}})\right)_j$ is a $t$-tail-deletion-correcting code.
By construction the code $\left(\insoperatorsphere{t}(\cC_{q,t}^{\mathsf{base}})\right)_j$ is of the same size as $\cC_{q,t}^{\mathsf{base}}$.

Let $1\leq j_1< j_2\leq t!$ and $\pi_i\in \left(\insoperatorsphere{t}(\cC_{q,t}^{\mathsf{base}})\right)_{j_i}$ for $i\in \{1,2\}$. By construction, there exist $x_1,x_2\in \cC_{q,t}^{\mathsf{base}}$ and $y_1,y_2\in \partperm{t}{q}$ such that $\pi_i=y_ix_i$ for $i\in \{1,2\}$. 
If $x_1=x_2$ then it holds that $\pi_i=(\insoperatorsphere{t}(x_1))_{j_i}$. Since $j_1\neq j_2$ it holds that $\pi_1\neq \pi_2$.
If $x_1\neq x_2$ there exist $y_1,y_2\in \partperm{t}{q}$ such that $\pi_i=y_ix_i$ for $i\in \{1,2\}$, and hence,
$\pi_1=y_1x_1\neq y_2x_2=\pi_2$. Therefore, the codes are disjoint.

Lastly, consider the case where $q\not\equiv 0\bmod (t+1)$.
 For every $\pi \in \partperm{1}{q}$ it holds that $\deloperator{t}(\pi)=\{\pi\}$. For any distinct $\pi_1,\pi_2\in \partperm{1}{q}$ it holds that $\deloperator{t}(\pi_1)\cap \deloperator{t}(\pi_2)\neq \phi$.
Furthermore, for any distinct $\pi_1,\pi_2\in \left(\insoperatorsphere{t}(\cC_{q,t}^{\mathsf{base}})\right)_{1}$ it holds that $\deloperator{t}(\pi_1)\cap \deloperator{t}(\pi_2)\neq \phi$ since it is a $t$-tail-deletion-correcting code.

 For every $\pi\in \left(\insoperatorsphere{t}(\cC_{q,t}^{\mathsf{base}})\right)_{1}$ it holds that $|\pi|\geq t+ (q-t)-{\left\lfloor\frac{q-t-1}{t+1}\right\rfloor}(t+1) > t+(q-t)-(q-t-1)=t+1$, while the second inequality derives from $q\not\equiv 0\bmod (t+1)$.
 Hence, $|\pi|\geq t+2$. Therefore, $\deloperator{t}(\pi)\cap \partperm{1}{q}=\phi$.
 In summary, for any distinct $\pi_1,\pi_2\in \left(\insoperatorsphere{t}(\cC_{q,t}^{\mathsf{base}})\right)_{1}\cup \partperm{1}{q}$ it holds that $\deloperator{t}(\pi_1)\cap \deloperator{t}(\pi_2)\neq \phi$, and therefore, the code $\left(\insoperatorsphere{t}(\cC_{q,t}^{\mathsf{base}})\right)_{1}\cup \partperm{1}{q}$ is a $t$-tail deletion-correcting code

\end{IEEEproof}

% The following theorem is an immediate results of the previous lemma.

\begin{theorem}
 Let $q$ be the size of the alphabet, and let $t< q$. If $q\equiv 0\bmod (t+1)$ it holds that  $\left(\insoperatorsphere{t}(\cC_{q,t}^{\mathsf{base}})\right)_{1}$ is an optimal $t$-tail-deletion-correcting code, i.e.,
 $
 \codesizecordel(q,t)=|\left(\insoperatorsphere{t}(\cC_{q,t}^{\mathsf{base}})\right)_{1}|=|\cC_{q,t}^{\mathsf{base}}|.
 $
% \begin{align*}
% \codesizecordel(q,t)&=\sum_{1\leq i\leq q-t : \ i\equiv q-t\bmod (t+1)}{|\partperm{i}{q}|}\\
%     &=\sum_{1\leq i\leq q-t : \ i\equiv q-t\bmod (t+1)}\frac{q!}{(q-i)!}.
% \end{align*}
% \begin{align*}
% \codesizecordel(q,t)&=\sum_{i=0}^{\left\lfloor\frac{q-t}{t+1}\right\rfloor}{|\partperm{q-t-i(t+1)}{q}|}=\sum_{i=0}^{\left\lfloor\frac{q-t}{t+1}\right\rfloor}\frac{q!}{(t+i(t+1))!}.
% \end{align*}
If $q\not\equiv 0\bmod (t+1)$ it holds that $\left(\insoperatorsphere{t}(\cC_{q,t}^{\mathsf{base}})\right)_{1}\cup\{1,2,3,\dots,q\}$ is an optimal $t$-tail-deletion-correcting code, i.e.,
$
 \codesizecordel(q,t)=|\cC_{q,t}^{\mathsf{base}}|+q.
 $
% \begin{align*}
% \codesizecordel(q,t)&=|\partperm{1}{q}|+\sum_{i=0}^{\left\lfloor\frac{q-t}{t+1}\right\rfloor}{|\partperm{q-t-i(t+1)}{q}|}\\
%     &=q+\sum_{i=0}^{\left\lfloor\frac{q-t}{t+1}\right\rfloor}\frac{q!}{(t+i(t+1))!}.
% \end{align*}
    
\end{theorem}

% \begin{theorem}
% Let $q$ be the size of the alphabet and let $t$ be the number of deletions we want to correct. If $q\equiv 0\bmod (t+1)$ it holds that
% \begin{align*}
% \codesizecordel(q,t)&=\sum_{1\leq i\leq q-t : \ i\equiv q-t\bmod (t+1)}{|\partperm{i}{q}|}\\
%     &=\sum_{1\leq i\leq q-t : \ i\equiv q-t\bmod (t+1)}\frac{q!}{(q-i)!}.
% \end{align*}
% \end{theorem}
\begin{IEEEproof}
% From Lemma~\ref{lm:ballinup} it holds that $\left(\cC_{q,t}^{\mathsf{base}}\uparrow_t\right)_{1}$ is a $t$-tail-deletion-correcting code. Furthermore,
% $|\left(\cC_{q,t}^{\mathsf{base}}\uparrow_t\right)_{1}|=|\cC_{q,t}^{\mathsf{base}}|=\sum_{i=0}^{\left\lfloor\frac{q-t-1}{t+1}\right\rfloor}{|\partperm{q-t-i(t+1)}{q}|}$.
First, we will bound the size of a $t$-tail-deletion-correcting code.
Let $\cC$ be a $t$-tail-deletion-correcting code.
 For every $1\leq i\leq q-t$ it holds that $|\{\cup_{0\leq j \leq t}{\partperm{i+j}{q}}\}\cap\cC|\leq |{\partperm{i}{q}}|$. Otherwise, using the pigeonhole principle there exist $\pi_1,\pi_2\in \cC$ such that $\deloperator{t}(\pi_1) \bigcap \deloperator{t}(\pi_2)\neq \phi.$

If $q\equiv 0\bmod (t+1)$ the sets $\{\cup_{0\leq j \leq t}{\partperm{q-t-i(t+1)+j}{q}}\}$ for $i\in \left\{0,1,2,\dots,{\left\lfloor\frac{q-t-1}{t+1}\right\rfloor}\right\}$
cover $\partpermallgen$, and
hence,
$
|\cC|\leq \sum_{i=0}^{\left\lfloor\frac{q-t-1}{t+1}\right\rfloor}{|\partperm{q-t-i(t+1)}{q}|}=|\cC_{q,t}^{\mathsf{base}}|.
$

If $q\not\equiv 0\bmod (t+1)$ the sets $\{\cup_{0\leq j \leq t}{\partperm{q-t-i(t+1)+j}{q}}\}$ for $i\in \left\{0,1,2,\dots,{\left\lfloor\frac{q-t-1}{t+1}\right\rfloor}\right\}$ and the set $\{\cup_{1\leq j \leq t+1}{\partperm{j}{q}}\}$
cover $\partpermallgen$, and
hence,
$
|\cC|\leq {|\partperm{1}{q}|}+\sum_{i=0}^{\left\lfloor\frac{q-t-1}{t+1}\right\rfloor}{|\partperm{q-t-i(t+1)}{q}|}=q+|\cC_{q,t}^{\mathsf{base}}|.
$

If $q\equiv 0\bmod (t+1)$ it holds that $|\left(\insoperatorsphere{t}(\cC_{q,t}^{\mathsf{base}})\right)_{1}|=|\cC_{q,t}^{\mathsf{base}}|$, and therefore, using Lemma~\ref{lm:ballinup} the code is optimal. If $q\not\equiv 0\bmod (t+1)$  it holds that $|\left(\insoperatorsphere{t}(\cC_{q,t}^{\mathsf{base}})\right)_{1}\cup \partperm{1}{q}|=|\cC_{q,t}^{\mathsf{base}}|+q$, and therefore, using again Lemma~\ref{lm:ballinup} the code is optimal.  

\end{IEEEproof}

\section{Codes over Vectors of Permutations}\label{sec:defvector}

In this section, we are interested in vectors of partial permutations, i.e., each one of the symbols in the vector is a partial permutation. This model fits the setup of using rank modulated DNA for storage, where in each position of the strand we are able to use a partial permutation. 

\subsection{Error Model and Problem Statement}
Let $\bfu,\bfv\in (\partpermallgen )^{n}$ be two vectors of partial permutations. It is said that $\bfu$ experienced \emph{$(t,e)$-tail deletions}, if for at most $e$ indices, the partial permutation $u_i$ experienced at most $t$ tail deletions from $v_i$.

\begin{example}
    For $q=5, n=2$ it holds that $(1345,135),(45,135)$ are vectors in $\{\partpermallgen\}^{n}$. It holds that $2$-tail deletions occurred from $1345$ to $45$.
    % \begin{align*}
    %     \distau(135, 351)=2.\\
    %     \distau(134, 134)=0.\\
    % \end{align*}
    Therefore, the vector $(45,135)$ experienced $(2,1)$-tail deletions from the vector $(1345,135)$. 
\end{example}

In the following section we call a subset of sequences $\cC\subset \partpermallgenseq$ a code. We say that a code over $\partpermallgenseq$ is a \emph{$(t,e)$-tail-deletion-detecting code} if it can detect any $(t,e)$-tail deletion. Similarly, a code over $\partpermallgenseq$ is a \emph{$(t,e)$-tail-deletion-correcting code} if it can correct any $(t,e)$-tail deletion. 
%Let $n$ be the length of the sequence, $\motifsize,\wordsize,t,e\in\mathbb{N}$.
Lastly, we denote by $\codesizedetdelseq(q,t,e), \codesizecordelseq(q,t,e)$, the largest length-$n$ $(t,e)$-tail-deletions detecting, correcting code, respectively. Codes that correct these types of deletions will be referred as \emph{tail deletion permutation codes}.

\begin{problem}\label{pb:codes}
    For every $n,q,t,e\in\mathbf{N}$, find the sizes of $\codesizedetdelseq(q,t,e)$ and $\codesizecordelseq(q,t,e)$ and find efficient constructions for such codes.
\end{problem}

\subsection{Constructions and Bounds}

% \section{Results for Tail Deletion Permutation Codes}

In this subsection, we present a similar construction of the construction of \emph{Tensor Permutation Codes} over partial permutations that was presented in~\cite{RMC25}. Furthermore, we present the same notations. 
In our paper the partition is over all partial permutations of any size and not only of the same size.

The idea behind this construction is based on the work of tensor product codes using two linear codes~\cite{TENS}, however, in our case the set of permutations does not constitute a linear space and thus cannot be easily used with the tensor product framework.

A family of sets ${\mathcal{R}={\mathcal{A}}_0,{\mathcal{A}}_1,\dots,{\mathcal{A}}_{\ell-1}\subseteq\partpermallgen}$ is called a \emph{partial partition} of size $\ell$ of $\partpermallgen$ if for every $i\neq j$ it holds that
$
{\mathcal{A}}_i\cap {\mathcal{A}}_j=\phi.
$ A partial partition $\mathcal{R}$
is called a \emph{partition} if 
$
\cup_{i=0}^{\ell-1}{{\mathcal{A}}_i}=\partpermallgen.
$
If for every $i$, the code ${\mathcal{A}}_i$ is a $t$-tail-deletion-detecting code then we say that the partial partition $\mathcal{R}$ is a \emph{$t$-tail-deletion-detecting partial partition}. If for every $i$ the code ${\mathcal{A}}_i$ is a $t$-tail-deletion-correcting code then we say that the partial partition $\mathcal{R}$ is a \emph{$t$-tail-deletion-correcting partial partition}.

For every vector of partial permutations $\bfc\in{\partpermallgenseq}$, we denote by $\Lambda(\bfc)$ the vector of indicators of the partial partition $\mathcal{R}$, i.e., in index $i$ we have the value $j$ if $\bfc_i\in {\mathcal{A}}_j$.  
If the partial permutation is not in any set we denote it by ``$?$''. We omit the partial partition $\mathcal{R}$ from the notation of $\Lambda$ as it will clear from the context. 

\begin{example}\label{ex:lambda}

With $q=4$ and $n=3$, let ${{\mathcal{A}}_0=\{123,21\},{\mathcal{A}}_1=\{1,231,24,2\}}$ be a partial partition of $\partpermallgenseq$.
    It holds that $\Lambda((123,123,3))=(0,0,?),\Lambda((123,21,21))=(0,0,0),\Lambda((321,24,21))=(?,1,0)$.    
\end{example}

Now we are ready to present the construction that follows the ideas of tensor product codes~\cite{TENS}. 

\begin{construction}\label{cnstr:TPC}
Let $\mathcal{R}=\{{\mathcal{A}}_0,{\mathcal{A}}_1,\dots,{\mathcal{A}}_{\ell-1}\}$ be a partial partition where ${\mathcal{A}}_i\subseteq\partpermallgen$ for $0\leq i\leq \ell-1$. This partition 
 will be referred as the \emph{inner code}. Let $\cC$ be a length-$n$ code over $\{0,1,2,\dots,\ell -1\}$ 
% \ey{are you sure the length is $n$?}
% with minimum Hamming distance $d_{\textbf{out}}$ 
% \ey{is it the same $q$ from the previous section? it shouldn't be}
that will be the \emph{outer code}. The tail tensor permutation code
$\mathbf{TTPC}(\mathcal{R},\cC)$ is defined as follows. A vector of partial permutations $\bfc\in({\partpermallgen})^n$ is a codeword in $\mathbf{TTPC}(\mathcal{R},\cC)$ if and only if 
$
\Lambda(\bfc)\in\cC,
$
i.e., 
$\mathbf{TTPC}(\mathcal{R},\cC)=\{\bfc\in({\partpermallgen})^n |\Lambda(\bfc)\in\cC\}$.
\end{construction}
A code that is constructed by Construction~\ref{cnstr:TPC} will be called a \emph{tail tensor permutation code}.
The following theorems can be proved using similar ideas as in~\cite{RMC25}. 
First, the following theorem calculates the size of tail tensor permutation codes.
\begin{theorem}\label{th:ttpcsize}
    Let ${\mathcal{R}}=\{{\mathcal{A}}_0,{\mathcal{A}}_1,\dots,{\mathcal{A}}_{\ell-1}\}$, where $ {\mathcal{A}}_i\subseteq\partpermallgen$ for $0\leq i\leq \ell-1$, 
be a partial partition and let $A\in\mathbf{N}$ such that for every $i$ it holds that $|{\mathcal{A}}_i|\geq A$. Let $\cC$ be a length-$n$ code over $\{0,1,2,\dots,\ell -1\}$. It holds that 
$
|\mathbf{TTPC}(\mathcal{R},\cC)|\geq A^n|\cC|.
$
\end{theorem}

The following theorem gives us the properties of tail tensor permutation codes. We assume that 
$\mathcal{R}=\{{\mathcal{A}}_0,{\mathcal{A}}_1,\dots,{\mathcal{A}}_{\ell-1}\}$ where ${\mathcal{A}}_i\subseteq\partpermallgen$ for $0\leq i\leq \ell-1$ is a partial partition and $\cC$ is a length-$n$ code over $\{0,1,2,\dots,\ell -1\}$. 
\begin{theorem}\label{th:ttpcdet}
If $\mathcal{R}$ is a $t$-tail-deletion-detecting partial partition and $\cC$ has Hamming distance $e+1$ then $\mathbf{TTPC}(\mathcal{R},\cC)$ is a $(t,e)$-tail-deletion-detecting code.
If $\mathcal{R}$ is a $t$-tail-correcting-detecting partial partition and $\cC$ has Hamming distance $2e+1$ then $\mathbf{TTPC}(\mathcal{R},\cC)$ is a $(t,e)$-tail-deletion-correcting code.
\end{theorem}

% \begin{theorem}\label{th:ttpccor}
% If $\mathcal{R}$ is a $t$-tail deletions correcting partial partition and $\cC$ has Hamming distance $2e+1$ then $\mathbf{TTPC}(\mathcal{R},\cC)$ is a $(t,e)$-tail-deletion-correcting code.
% \end{theorem}

% \begin{theorem}
% Let ${\mathcal{R}=A_0,A_1,\dots,A_{\ell-1}\subseteq\partpermallgen}$ be a partial partition and let $\cC$ be a length-$n$ code over $\{0,1,2,\dots,\ell -1\}$.
% If $\mathcal{R}$ can correct $d$-tail deletions and $\cC$ has Hamming distance $e$ then $\mathbf{TTPC}(\mathcal{R},\cC)$ is a $(d,\frac{e-1}{2})$-tail deletions correcting code.
% \end{theorem}

In the following theorems, we present explicit constructions for tail tensor permutation codes using partial permutations. We denote by $A_q(n,d)$ the size of a largest code over $[q]$ of size $n$ with minimum Hamming distance $d$.

% or the following theorem we expand the definition of $\pi\uparrow_j$.

\begin{theorem}
For every $n,q,t,e\in\mathbf{N}$ such that $t<q$, it holds that there exists a code that can correct $(t,e)$-tail deletions. 
    Furthermore, it holds that,
    \[
    \codesizecordelseq(q,t,e)\geq {|\cC_{q,t}^{\mathsf{base}}|}^{n}A_{t!}(n,2e+1).
    \]
    Moreover, if $q\equiv 0\bmod (t+1)$ it holds that
    \[
    \codesizecordelseq(q,t,e)\geq \codesizecordel(q,t)^{n}A_{t!}(n,2e+1).
    \]
\end{theorem}
\begin{IEEEproof}
As defined before, denote,
% $
% \cC_{q,t}^{\mathsf{base}}=\bigcup_{1\leq i\leq q-d : \ i\equiv q-d\bmod (t+1)}\partperm{i}{q}.
% $
$
\cC_{q,t}^{\mathsf{base}}=\bigcup_{i=0}^{\left\lfloor\frac{q-t-1}{t+1}\right\rfloor}
\partperm{q-t-i(t+1)}{q}.
$
Using Lemma~\ref{lm:ballinup} for every $1\leq j\leq t!$ it holds that the code $\left(\insoperatorsphere{t}(\cC_{q,t}^{\mathsf{base}})\right)_j$ is a $t$-tail-deletion-correcting code. Furthermore, for every $j_1,j_2$ it holds that 
$
|\left(\insoperatorsphere{t}(\cC_{q,t}^{\mathsf{base}})\right)_{j_1}|=|\left(\insoperatorsphere{t}(\cC_{q,t}^{\mathsf{base}})\right)_{j_2}|
$ and the codes are disjoint.
Hence, the partial partition $\mathcal{R}=\left\{\left(\insoperatorsphere{t}(\cC_{q,t}^{\mathsf{base}})\right)_1,\left(\insoperatorsphere{t}(\cC_{q,t}^{\mathsf{base}})\right)_2,\dots,\left(\insoperatorsphere{t}(\cC_{q,t}^{\mathsf{base}})\right)_{t!}\right\}$ is a $t$-tail-deletion-correcting partial partition. Let $\cC$ be a largest length-$n$ code over $\{0,1,2,\dots,t! -1\}$ with Hamming distance $2e+1$. It holds from Theorem~\ref{th:ttpcdet} that the code $\mathbf{TTPC}\left(\mathcal{R},\cC\right)$ is a $(t,e)$-tail-deletion-correcting code. Furthermore, from Theorem~\ref{th:ttpcsize} it holds that $
|\mathbf{TTPC}\left(\mathcal{R},\cC\right)|=|\cC_{q,t}^{\mathsf{base}}|^n A_{t!}(n,2e+1).$ Hence,
\[
\codesizecordelseq(q,t,e) \geq |\cC_{q,t}^{\mathsf{base}}|^n A_{t!}(n,2e+1).\]
Moreover, if $q\equiv 0\bmod (t+1)$ it holds that $\codesizecordel(q,t)={|\cC_{q,t}^{\mathsf{base}}|}$, hence,
    \[
    \codesizecordelseq(q,t,e) \geq\codesizecordel(q,t)^n A_{t!}(n,2e+1)
    \]
\end{IEEEproof}

\begin{theorem}
For every $n,q,t,e\in\mathbf{N}$ such that $t<q$  it holds that there exists a code that can detect $(t,e)$-tail deletions. 
    Furthermore,
    \[
    \codesizedetdelseq(q,t,e)\geq {|\cC_{q,t}^{\mathsf{base}}|}^{n} A_{t+1}(n,e+1).
    \]
    Moreover, if $q\equiv 0\bmod (t+1)$ it holds that
    \[
    \codesizedetdelseq(q,t,e)\geq {\codesizecordel(q,t)}^{n} A_{t+1}(n,e+1).
    \]
    \end{theorem}
\begin{IEEEproof}
For any $0\leq j\leq t$, denote
    % \[
    % \cC_j=\cup_{1\leq i\leq q-j : \ i\equiv q-j\bmod (t+1)}\{\partperm{i}{q}\}.
    % \]
    $\cC_j=\bigcup_{i=0}^{\left\lfloor\frac{q-j-1}{t+1}\right\rfloor}\partperm{q-j-i(t+1)}{q}$.
    One can note that the codes are disjoint, furthermore, $\cC_t$ has the smallest size. Similarly to the proof of Theorem~\ref{th:bestdetperm} for any $0\leq j\leq t$ it holds that $\cC_j$ is a $t$-tail-deletion-detecting code. Let $\cC$ be a largest length-$n$ code over $\{0,1,2,\dots,t\}$ with Hamming distance $e+1$. It holds from Theorem~\ref{th:ttpcdet} that the code $\mathbf{TTPC}(\{\cC_0,\cC_1,\dots,\cC_{t}\},\cC)$ is a $(t,e)$-tail-deletion-detecting code.
    One can note that for $i<j$ it holds that $|\cC_i|\geq|\cC_j|$, and furthermore, $\cC_t= \cC_{q,t}^{\mathsf{base}}$, hence, for every $i$ it holds that 
    $
    |\cC_i|\geq |\cC_{q,t}^{\mathsf{base}}|.
    $
    Using Theorem~\ref{th:ttpcsize} it holds that $|\mathbf{TTPC}(\{\cC_0,\cC_1,\dots,\cC_t\},\cC)|\geq |\cC_t|^nA_{t+1}(n,e+1)={|\cC_{q,t}^{\mathsf{base}}|}^{n}A_{t+1}(n,e+1)$. 
    Hence,
    \[
    \codesizedetdelseq(q,t,e)\geq {|\cC_{q,t}^{\mathsf{base}}|}^{n} A_{t+1}(n,e+1).
    \]
    Moreover, if $q\equiv 0\bmod (t+1)$ it holds that $\codesizecordel(q,t)={|\cC_{q,t}^{\mathsf{base}}|}$, hence,
    \[
    \codesizedetdelseq(q,t,e)\geq {\codesizecordel(q,t)}^{n} A_{t+1}(n,e+1).
    \]
\end{IEEEproof}

% \subsection{PDF Requirements}

% Only electronic submissions in form of a PDF file will be
% accepted. The PDF file has to be PDF/A compliant. A common problem is
% missing fonts. Make sure that all fonts are embedded. (In some cases,
% printing a PDF to a PostScript file, and then creating a new PDF with
% Acrobat Distiller, may do the trick.) More information (including
% suitable Acrobat Distiller Settings) is available from the IEEE
% website \cite{IEEE:pdfsettings, IEEE:AuthorToolbox}.

% \section{Conclusion}

% We conclude by pointing out that on the last page the columns need to
% balanced. Instructions for that purpose are given in the source file.

% Moreover, example code for an appendix (or appendices) can also be
% found in the source file (they are commented out).

%%%%%%
%% Appendix:
%% If needed a single appendix is created by
%%
%\appendix
%%
%% If several appendices are needed, then the command
%%
% \appendices
%%
%% in combination with further \section commands can be used.
%%%%%%

\section*{Acknowledgments}
This work was funded by the European Union (DiDAX, 101115134). Views and opinions expressed are however those of the authors only and do not necessarily reflect those of the European Union or the European Research Council Executive Agency. Neither the European Union nor the granting authority can be held responsible for them.

 The collaboration and discussions of this paper were originated in Dagstuhl Seminar 24511~\cite{DSTU}.


\begin{thebibliography}{9}

\bibitem{anavy_DataStorageDNA_2019} 
 L. Anavy, I. Vaknin, O. Atar, R. Amit, and Z. Yakhini, 
``Data storage in DNA with fewer synthesis cycles using composite DNA letters," \emph{Nature Biotechnology}, vol. 37, no. 10, pp. 1229–1236, 2019.


\bibitem{BAMA}
A. Barg and A. Mazumdar, "Codes in permutations and error correction for rank modulation," in \emph{IEEE Trans. on Inf. Theory, vol. 56, no. 7}, pp. 3158-3165, July 2010.

\bibitem{augmented_encoding}
Y. Choi et al., 
%, T. Ryu,  A. Lee, H. Choi,  H. Lee,  J. Park,  S.H., Song,  S. Kim,  H. Kim,  W. Park,  and S. Kwon, 
``High information capacity DNA-based data storage with augmented encoding characters using degenerate bases,"
\emph{Scientific Reports}, vol.9, 2019. 

\bibitem{TC24}
T. Cohen and E. Yaakobi, "Optimizing the decoding probability and coverage ratio of composite DNA," \emph{IEEE Int. Symp. Inf. Theory (ISIT), Athens, Greece, 2024}, pp. 1949-1954. 

\bibitem{RMC25}
T. Cohen, Z. Wang, E. Yaakobi and Z. Yakhini, "Rank modulated composite encoding for data storage in DNA," \emph{2025 13th International Symposium on Topics in Coding (ISTC)}, Los Angeles, CA, USA, 2025, pp. 1-5.

\bibitem{RMEitanRyan}
R. Gabrys, E. Yaakobi, F. Farnoud and J. Bruck, "Codes correcting erasures and deletions for rank modulation," \emph{IEEE International Symposium on Information Theory}, Honolulu, HI, USA, 2014, pp. 2759-2763.

\bibitem{Getal13} 
N. Goldman et al., %, P. Bertone, S. Chen, C. Dessimoz, E. M. LeProust, B. Sipos, and E. Birney, 
``Towards practical, high-capacity, low-maintenance information storage in synthesized DNA," \emph{Nature}, vol. 494, no. 7435, pp. 77--80, 2013.

\bibitem{RAMOSCH}
A. Jiang, R. Mateescu, M. Schwartz and J. Bruck, "Rank modulation for flash memories," in \emph{IEEE Trans. on Inf. Theory, vol. 55, no. 6}, pp. 2659-2673, June 2009.

\bibitem{RobRank}
A. Jiang, R. Mateescu, M. Schwartz, and J. Bruck, “Rank modulation for
flash memories,” \emph{IEEE Trans. Inf. Theory, vol. 55, no. 6}, pp. 2659–2673,
Jun. 2009.

\bibitem{CORRSCH}
A. Jiang, M. Schwartz and J. Bruck, "Correcting charge-constrained errors in the rank-modulation scheme," in \emph{IEEE Trans. on Inf. Theory, vol. 56, no. 5}, pp. 2112-2120, May 2010.



\bibitem{KNUTH}
D. E. Knuth, \emph{The Art of Computer Programming Volume 3: Sorting and Searching}, 2nd ed. Reading, MA: Addison-Wesley, 1998.

\bibitem{KYW23}
A. Kobovich, E. Yaakobi, and N. Weinberger, 
``M-DAB: An input-distribution optimization algorithm for composite DNA storage by the multinomial channel," Arxiv, Sep., 2023.

\bibitem{osti_1619517}
H. Lee, R. Kalhor, N. Goela, J. Bolot, and G.M. Church,
``Terminator-free template-independent enzymatic DNA synthesis for digital information storage,"
\emph{Nature Communications}, vol. 10, no. 1, Jun., 2019.

\bibitem{Levenshtein1966}
V. Levenshtein,
"Binary codes capable of correcting deletions, insertions, and reversals,"
\emph{Soviet Physics Doklady}, vol.~10, no.~8, pp.~707--710, 1966.


% \bibitem{ISIT24_1}
% O. Sabary, I. Preuss, R. Gabrys, Z. Yakhini, L. Anavy, and E. Yaakobi, 
% ``Error-correcting codes for combinatorial DNA composite," submitted to \emph{IEEE International Symposium on Information Theory}, 2024.


\bibitem{PGYA24}
I. Preuss, B. Galili, Z. Yakhini, and L. Anavy, ``Sequencing coverage analysis for combinatorial DNA-based storage systems," BioRxiv Jan., 2024, https://www.biorxiv.org/content/10.1101/
2024.01.10.574966v1.

\bibitem{OM1}
O. Sabary, I. Preuss, R. Gabrys, Z. Yakhini, L. Anavy and E. Yaakobi, "Error-correcting codes for combinatorial composite DNA,"  \emph{IEEE Int. Symp. Inf. Theory (ISIT), Athens, Greece, 2024, pp}. 109-114.

\bibitem{ROM}
R. Sokolovskii, P. Agarwal, L. A. Croquevielle, Z. Zhou and T. Heinis, "Coding over coupon collector channels for combinatorial motif-based DNA storage," in \emph{IEEE Trans. on Comm.}, Early Access.


\bibitem{DSTU}
R. B., Olgica Milenkovic, Zohar Yakhini, Yonatan Yehezkeally, Anisha Banerjee, and Frederik Walter. Coding Theory and Algorithms for Emerging Technologies in Synthetic Biology (Dagstuhl Seminar 24511). In Dagstuhl Reports, Volume 14, Issue 12, pp. 46-62, Schloss Dagstuhl – Leibniz-Zentrum für Informatik (2025) https://doi.org/10.4230/DagRep.14.12.46

% \bibitem{PRYA24}
% I. Preuss, M. Rosenberg, Z. Yakhini, and L. Anavy, ``Efficient DNA-based data storage using shortmer combinatorial encoding," (Jan. 2,
% 2024), [Online]. Available: https://www.biorxiv.org/content/10.1101/
% 2021.08.01.454622v2 (visited on 01/29/2024), preprint.

% \bibitem{ISIT24_2}
% F. Walter, O. Sabary, A. Wachter-Zeh, and E. Yaakobi, 
% ``Coding for composite DNA to correct substitutions, strand losses, and deletions," submitted to \emph{IEEE International Symposium on Information Theory}, 2024.

\bibitem{OM2}
F. Walter, O. Sabary, A. Wachter-Zeh and E. Yaakobi, "Coding for composite DNA to correct substitutions, Strand Losses, and Deletions," \emph{IEEE Int. Symp. Inf. Theory (ISIT), Athens, Greece, 2024, pp}. 97-102.


\bibitem{WESEL}
R. D. Wesel, E. E. Wesel, L. Vandenberghe, C. Komninakis and M. Medard, "Efficient Binomial Channel Capacity Computation with an Application to Molecular Communication," 2018 Information Theory and Applications Workshop (ITA), San Diego, CA, USA, 2018, pp. 1-5, doi: 10.1109/ITA.2018.8503225.
\bibitem{TENS}
J. Wolf, "On codes derivable from the tensor product of check matrices," in \emph{IEEE Trans. on Inf. Theory, vol. 11, no. 2, pp}. 281-284, April 1965.

\bibitem{ZC22}
W. Zhang, Z. Chen, and Z. Wang, 
``Limited-magnitude error correction for probability vectors in DNA storage,"
\emph{ IEEE International Conference on Communications (ICC)}, pp. 3460--3465, 2022. 

\end{thebibliography}
\end{document}